\documentclass{article}
\usepackage[affil-it]{authblk}

\usepackage[english]{babel}
\usepackage{mathrsfs}
\usepackage{amsmath, amssymb}
\usepackage{slashed}
\usepackage{nicefrac}

\pdfoutput=1
\usepackage{hyperref}

\usepackage{tabularx}
\usepackage{booktabs}

\usepackage{amsthm}

	\theoremstyle{plain}
		\newtheorem{theorem}{Theorem}
		\newtheorem{lem}[theorem]{Lemma}
		\newtheorem{prop}[theorem]{Proposition}

	\theoremstyle{plain}
		
		\newtheorem{rmk}[theorem]{Remark}
		\newtheorem{exmpl}{Example}

\usepackage{a4wide}
\newcommand{\rep}[2]{
	$\mathbf{#1} \otimes \mathbf{#2}^o$
}

\def\A{\mathcal{A}}
\def\H{\mathcal{H}}
\def\com{\mathbb{C}}
\def\F{\mathbb{F}}
\def\R{\mathbb{R}}
\def\C{\com}
\def\Qu{\mathbb{H}}
\def\Z{\mathbb{Z}}

\DeclareMathOperator{\End}{End}
\DeclareMathOperator{\tr}{tr}
\DeclareMathOperator{\id}{id}
\DeclareMathOperator{\ad}{ad}
\DeclareMathOperator{\diag}{diag}

\title{Going beyond the Standard Model \\with noncommutative geometry}

\author{Thijs van den Broek%
  \thanks{Electronic address: \texttt{T.vandenBroek@science.ru.nl}}}
\affil{Nikhef, Science Park Amsterdam 105, 1098 XG Amsterdam / Radboud University Nijmegen} 

\author{Walter D.~van Suijlekom%
  \thanks{Electronic address: \texttt{waltervs@math.ru.nl}}}
\affil{Radboud University Nijmegen, Institute for Mathematics, Astrophysics and Particle Physics,
Faculty of Science, Heyendaalseweg 135, 6525 AJ Nijmegen, The Netherlands}

\date{\today}

\begin{document}
\maketitle

\begin{abstract}
The derivation of the full Standard Model from noncommutative geometry has been a promising sign for possible applications of the latter in High Energy Physics. Many believe, however, that the Standard Model cannot be the final answer. We translate several demands whose origin lies in physics to the context of noncommutative geometry and use these to put constraints on the fermionic content of models. We show that the Standard Model only satisfies these demands provided it has a right-handed neutrino in each `generation'. Furthermore, we show that the demands can be met upon extending the SM with a copy of the representation $\mathbf{2} \otimes \overline{\mathbf{1}}^o$, but this has consequences for the number of particle generations. We finally prove that the Minimal Supersymmetric Standard Model is not among the models that satisfy our constraints, but we pose a solution that is a slight extension of the MSSM.
\end{abstract}

\section{Introduction}

The Standard Model (SM) of elementary particles has been tremendously successful in explaining the world at the smallest length scales that can currently be probed \cite{SSZ02}. Yet, it leaves many feeling a bit uneasy for some of its properties appear to be rather ad hoc; few people believe that we have fully \emph{understood} the Standard Model. The application of noncommutative geometry \cite{C94} (NCG) to the subatomic realm might over time increase our understanding of the Standard Model. A line of thought that started with the Connes-Lott model \cite{CL89} culminated in a geometric description \cite{CCM} of the full Standard Model. It derives both the correct particle content of the Standard Model, extended with right-handed neutrinos ($\nu$SM), as its full action by employing the principles of NCG and adding as little extra input as possible. On top of that, this description allows for a prediction of the Higgs boson mass \cite{CCM, CC12}. NCG might thus prove itself a valuable tool for model building in High Energy Physics. \\

Considerable effort has already been spent on a classification of all the possible models using the demands that various mathematical and physical arguments put on noncommutative geometries, such as \cite{ISS03,JS05,JSS05,JS08,JS09} and also \cite{CC08}. The approach we take in this article is to exploit some of the more recent developments (see \cite{CCM}) to put constraints on all possible SM extensions. We will demand from any model that
	\begin{enumerate}
		\item the gauge Lie algebra associated to the noncommutative geometry is that of Standard Model (\S \ref{sec:alg});
		\item the particle content contains at least one copy of each of the particles that the Standard Model features (\S \ref{sec:alg});
		\item the hypercharges of the particles are such that there is no chiral gauge anomaly (\S \ref{sec:anomalies});
		\item the values of the coupling constants of the electromagnetic, weak and strong forces are such that they satisfy the GUT-relation\footnote{This relation not only appears in the context of $SU(5)$ Grand Unified Theory, but is also a feature of $SO(10)$ and $E_6$ theories \cite{PL81}, \cite{CEG77}.}
			\begin{align}
			\frac{5}{3}g_1^2 = g_2^2= g_3^2\label{eq:gutrel}
			\end{align}
			(\S \ref{sec:gut}). It is in fact this relation ---coming out naturally in the description of the SM from NCG--- that allows one to determine a scale on which the theory `lives'. Consequently it plays a vital role in the prediction for the value of the Higgs mass. 
	\end{enumerate}
Together these four demands lead to a number of relations between the multiplicities of the particles, which can be used to constrain the number of viable models. We recover the Standard Model (plus a right-handed neutrino in each generation) using these relations in \S \ref{sec:sm} and finally turn our attention to supersymmetric variants (\S \ref{sec:susy}). 
We must add however, that the scope of the application of these results is much broader than supersymmetry alone.

\section{Preliminaries}

At the very heart of NCG lies the notion of a \emph{spectral triple}, describing a \emph{noncommutative manifold}. It is a triple $(\A, \H, D)$, where $\A$ is a unital $*$-algebra that is represented as bounded operators on a Hilbert space $\H$ on which a \emph{Dirac operator} $D$ acts. The latter is an (unbounded) self-adjoint operator that has compact resolvent and in addition satisfies $[D, a] \in B(\H)\ \forall\ a \in \A$.  
\begin{itemize}
\item We call a spectral triple \emph{even} if there exists a grading $\gamma : \H\to\H$, with $[\gamma, a] = 0\ \forall\ a \in \A$ and $\gamma D = - D\gamma$.
\item We call a spectral triple \emph{real} if there exists an anti-unitary \emph{real structure} $J : \H \to \H$ satisfying
\begin{align*}
	J^2 & = \pm 1, & JD &= \pm DJ. 
\end{align*}
The Dirac operator and real structure are required to be compatible via the \emph{first-order condition}
	$
		[[D, a], Jb^*J^*] = 0\ \forall\ a, b \in \A.
	$
\item If a spectral triple is both real and even there is the additional compatibility relation 
\begin{align*}
	J\gamma = \pm \gamma J.
\end{align*}
The eight different combinations for the tree signs above determine the \emph{KO-dimension} of the spectral triple. For more details we refer to \cite{CM}.
\end{itemize}

This is a rather abstract notion which we will try to make more concrete by providing two examples that play a key role in the application of NCG to particle physics. 

\begin{exmpl}{(Canonical spectral triple)}\label{ex:canon}
	The triple 
	$
		(\A, \H, D) = (C^{\infty}(M), L^2(M, S), \slashed{\partial} := i\slashed{\nabla}^S)
	$
	serves as the motivating example of a spectral triple. Here $M$ is a compact Riemannian manifold that has a spin structure and $L^2(M, S)$ denotes the square-integrable sections of the corresponding spinor bundle. The Dirac operator $\slashed{\partial}$ is associated to the unique spin connection which in turn is derived from the Levi-Civita connection on $M$. This spectral triple can be extended by a real structure $J_M$ (`charge conjugation') and ---when $\dim M$ is even--- a grading $\gamma_M \equiv \gamma^{\dim M + 1}$ (`chirality'). The KO-dimension of a canonical spectral triple equals the dimension of $M$.	
\end{exmpl}

\subsection{Finite spectral triples}

A second example of spectral triples is that of a \emph{finite spectral triple}:

\begin{exmpl}{(Finite spectral triple \cite{PS96, KR97})}\label{ex:finite}
	For a finite-dimensional algebra $\A$, a finite-dimensional (left) module $\H$ of $\A$ and a symmetric linear operator $D : \H \to \H$, we call $(\A, \H, D)$ a \emph{finite spectral triple}. As in the general case some finite spectral triples are called real and/or even depending on whether there exists a $J$ (implementing a bimodule structure of $\H$) and/or $\gamma$ respectively. 
\end{exmpl}
To be more specific, we will write out in detail the properties of such finite spectral triples, restricting to the case of $KO$-dimension $6$ for which $J^2 =1, JD = DJ$ and $J \gamma = - \gamma J$:
\begin{itemize}
	\item[-] The finite-dimensional algebra is (by Wedderburn's Theorem) a direct sum of matrix algebras:
	\begin{align}
		\A = \bigoplus_{i}^K M_{N_i}(\mathbb{F}_i)\qquad \mathbb{F}_i = \mathbb{R}, \mathbb{C}, \mathbb{H}\label{eq:algebra},
	\end{align}	
	with componentwise multiplication: $(a_1, \ldots, a_K)(b_1, \ldots, b_K) = (a_1b_1, \ldots, a_Kb_K)$.
	
	\item[-] The Hilbert space is an $\A\otimes (\A)^o$-module, or, more specifically, a direct sum of tensor products of two defining representations $\mathbf{N}_i$ of $M_{N_i}(\mathbb{F}_i)$:
	\begin{align}
		\H_F = \bigoplus_{i \leq j\leq K} \Big[\mathbf{N}_i \otimes \mathbf{N}_j^o \oplus \mathbf{N}_j \otimes \mathbf{N}_i^o\Big]^{\oplus M_{N_iN_j}} 
\label{eq:Hilbertspace}.
	\end{align}
The non-negative integers $M_{N_iN_j}$ denote the \emph{multiplicity} of the representation $\mathbf{N}_i \otimes \mathbf{N}_j^o$.
	\item[-] The $\A$-bimodule structure is implemented by a real structure $J$ interchanging representations $\mathbf{N}_i \otimes \mathbf{N}_j^o$ and $\mathbf{N}_j \otimes \mathbf{N}_i^o$ (so $M_{N_iN_j} = M_{N_jN_i}$).
	\item[-] For each component of the algebra for which $\mathbb{F}_i = \mathbb{C}$ we will a priori allow both the complex linear representation $\mathbf{N}_i$ as the real-linear representation $\overline{\mathbf{N}}_i$ given by
	\begin{align*}
		\pi' &: M_{N_i}(\mathbb{F}_i) \to \End_{\mathbb{R}}(\mathbb{F}^{N_i}), & 
		\pi'(m)v &:= \overline{m}v,\qquad v \in \mathbb{F}^N.
	\end{align*}
\end{itemize}

The algebra \eqref{eq:algebra} is a \emph{real} algebra. We are using this general form since already in describing the Standard Model (let alone extensions of it) one uses a copy of the quaterions $\mathbb{H}$, a real algebra.

\subsection{Almost-commutative geometries} 

In describing particle models using NCG one constructs spectral triples by taking the tensor product of a canonical spectral triple (Example \ref{ex:canon}) and a finite spectral triple (Example \ref{ex:finite}). Spectral triples of this form are generally referred to as \emph{almost-commutative geometries}. This tensor product simply consists of taking the tensor product of the various components of the two spectral triples, except for the Dirac operator, which is given by $\slashed{\partial} \otimes 1 + \gamma_M \otimes D_F$, where $D_F$ is the Dirac operator of the finite spectral triple. In this article we will not deal with operators $D_F$, in fact at some points we will even demand it to be absent by putting $D_F = 0$. If confusion is likely to arise, we will add a subscript $F$ to distinguish (a component of) a finite spectral triple from its canonical counterpart.\\

In the canonical spectral triple the algebra encodes space(time), in a finite spectral triple it is intimately connected to the gauge group. Noncommutative geometry can thus be said to put the external and internal degrees of freedom of a particle on similar footing.\\

When part of an almost-commutative geometry, the Dirac operator of the canonical spectral triple develops extra terms that have the characteristics of gauge fields \cite[\S \MakeUppercase{\romannumeral 11}]{C00}, whereas the finite Dirac operator $D_F$ can develop extra terms that parametrize scalar fields, such as the Higgs field. We will write $D_A := \slashed{\partial} \otimes 1 + \mathbb{A} + \gamma_M \otimes D_F(\phi)$, where the symbol $\mathbb{A}$ represents the generic gauge fields and the scalar fields are generically denoted by $\phi$.\\ 

With an almost-commutative geometry we can associate \cite{CC96} an \emph{action} functional (consisting of two terms: the {\it fermionic action} and the {\it spectral action}) via 
\begin{align}
	S[\bar\zeta, \zeta, A] := \frac{1}{2}\langle J\zeta, D_A \zeta \rangle + \tr f(D_A^2/\Lambda^2) \label{eq:totalaction}
\end{align}	
where $f$ is a positive function and
\begin{align}
	\langle J\zeta, D_A \zeta \rangle = \int_M (J\zeta, D_A\zeta)\sqrt{g}\mathrm{d}^4x,\qquad \zeta \in \frac{1}{2}(1 + \gamma_M \otimes \gamma_F)\H \equiv \H^+,\label{eq:fermionaction}
\end{align}
with $(.,.) :\H \times \H \to C^{\infty}(M)$ the Hermitian structure on $L^2(M, S)$ combined with the inner product in $\H_F$.\footnote{Note that, because of the form of an almost-commutative spectral triple, the requirement that $\gamma\zeta = \zeta$ does not imply that $\gamma_M \zeta = \zeta$.} \\

In practice one has to resort to approximations for calculating the second part of \eqref{eq:totalaction} explicitly. Most often this is done \cite{CM97} via a \emph{heat kernel expansion} \cite{GIL84}. For convenience, we end this section by recalling some details on this that will be of value later on. If $V$ is a vector bundle on a compact Riemannian manifold $(M, g)$ and if $P:C^{\infty}(V)\to C^{\infty}(V)$ is a second-order elliptic differential operator of the form
\begin{equation}
 P = - \big(g^{\mu\nu}\partial_{\mu}\partial_{\nu} + K^{\mu}\partial_{\mu} + L)\label{eq:elliptic} 
\end{equation}
with $K^{\mu}, L \in \Gamma(\End(V))$, then there exist a unique connection $\nabla$ and an endomorphism $E$ on $V$ such that
\begin{align*}
  P = \nabla\nabla^* - E.
\end{align*}
Explicitly, we write locally $\nabla_{\mu} = \partial_{\mu} + \omega'_{\mu}$, where $\omega_\mu' = \frac{1}{2}(g_{\mu\nu}K_\nu + g_{\mu\nu}g^{\rho\sigma}\Gamma_{\rho\sigma}^{\nu})$. Using this $\omega'_\mu$ and $L$ we find $E \in \Gamma(\End(V))$ and compute for the curvature $\Omega_{\mu\nu}$ of $\nabla$:
\begin{subequations}
\begin{align*}
  E &:= L - g^{\mu\nu}\partial_{\nu}(\omega_{\mu}') - g^{\mu\nu}\omega_\mu'\omega_\nu' + g^{\mu\nu}\omega_\rho'\Gamma_{\mu\nu}^\rho; \\
 \Omega_{\mu\nu} &:= \partial_{\mu}(\omega'_{\nu}) - \partial_{\nu}(\omega'_{\mu}) + [\omega'_{\mu},\omega'_{\nu}].
\end{align*}
\end{subequations}
It is this context in which one is able to compute the spectral action, by making an asymptotic expansion (as $t \to 0$) of the trace of the operator $e^{-tP}$ in powers of $t$:
\begin{align*}
\tr\,e^{-tP} &\sim \sum_{n \geq 0}t^{(n-m)/2}\int_M\tr e_n(x, P)\sqrt{g}\mathrm{d}^mx,
\end{align*}
where $m$ is the dimension of $M$ and the coefficients $\tr e_n(x, P)$ are called the \emph{Seeley--DeWitt coefficients}. One finds \cite[Ch~4.8]{GIL84} that for $n$ odd, $e_n(x, P) = 0$ and the first three even coefficients are given by
\begin{subequations}
	\begin{align}
		e_0(x, P) &= (4\pi)^{-m/2}(\id)\label{eq:gilkey1};\\
		e_2(x, P) &= (4\pi)^{-m/2}(-R/6\,\id + E)\label{eq:gilkey2};\\
		e_4(x, P) &= (4\pi)^{-m/2}\frac{1}{360}\big(-12R_{;\mu}^{\phantom{;\mu}\mu} + 5R^2 - 2R_{\mu\nu}R^{\mu\nu} \label{eq:gilkey3} \\
 &\qquad+ 2R_{\mu\nu\rho\sigma}R^{\mu\nu\rho\sigma} - 60RE + 180E^2 +60 E_{;\mu}^{\phantom{;\mu}\mu} + 30\Omega_{\mu\nu}\Omega^{\mu\nu}\big) \nonumber ,
	\end{align}
\end{subequations}
where $R$ is the scalar curvature of $M$, $R_{;\mu}^{\phantom{;\mu}\mu} := \nabla^{\mu}\nabla_{\mu}R$ and the same for $E$. \emph{In all cases that we will consider, the manifold $M$ will be taken four-dimensional and without boundary} so that the terms $E_{;\mu}^{\phantom{;\mu}\mu}, R_{;\mu}^{\phantom{;\mu}\mu}$ vanish by Stokes' Theorem.\\

Applying the heat kernel expansion of the second term in the action \eqref{eq:totalaction} then gives
\begin{align}
	\tr f(D_A^2/\Lambda^2) &\sim  2\Lambda^4 f_4 \tr e_0(D) + 2f_2\Lambda^2 \tr e_2(D) + f(0) \tr e_4(D) + \mathcal{O}(\Lambda^{-2})\label{eq:total_action_exp},
\end{align}
where $f_i$ is the $(i-1)$th moment of $f$:
\begin{align*}
	f_i = \int_0^{\infty} f(x) x^{i - 1}\mathrm{d}x.
\end{align*}

We underline that the action \eqref{eq:total_action_exp} is automatically gauge invariant and is fixed by the contents of the spectral triple.

\section{Constraints on the algebra $\A$}\label{sec:alg}
Let $\A$ be a $*$-algebra that is represented on a Hilbert space $\H$. Corresponding to the pair $(\A, \H)$ we define the \emph{gauge group}
	\begin{align*}
		SU(\A) := \{u \in U(\A), \det{}_{\H}(u) = 1\} 
	\end{align*}
where $U(\A)$ is the group of elements $u \in \A$ such that $uu^* = u^*u = 1$ and with $\det{}_{\H}$ we mean the determinant of the representation $\pi : \A \to B(\H)$. We then have
\begin{lem}
	Suppose that $\A$ is such that $su(\A) \simeq su(3) \oplus su(2) \oplus u(1)$. Then $\A$ must be in the following list:
	\begin{enumerate}
		\item $M_3(\com) \oplus \Qu \oplus M_2(\R)$,
		\item $M_3(\com) \oplus M_3(\mathbb{R}) \oplus M_2(\mathbb{R})$,
		\item $M_3(\com) \oplus \Qu \oplus \C$,
		\item $M_3(\com) \oplus M_2(\C)$ or
		\item $M_3(\com) \oplus M_3(\mathbb{R}) \oplus \com$,
	\end{enumerate}
	modulo extra summands $\mathbb{R}$.
\end{lem}
\begin{proof}
Let $\A$ be of the form \eqref{eq:algebra}, represented on a Hilbert space $\H$. We define two Lie algebras
\begin{align}
u(\A) &=  \{ X \in \A : X^* = - X \}\nonumber\\
su(\A) &= \{ X \in u(\A) : \tr_\H X = 0 \}\label{eq:lie-alg}
\end{align}
Note that thus $u(\A)$ is a direct sum of simple Lie algebras $o(N_i), u(N_i), sp(N_i)$ according to $\F_i = \R,\C,\Qu$, respectively. All these matrix Lie algebras have a trace, and we observe that the matrices in $o(N_i)$ and $sp(N_i)$ are already traceless. For the complex case, we can write $X_i \in u(N_i)$ as $X_i= Y_i+z_i$ where $z_i = \tr X_i$, showing that:
$$
u(N_i) = su(N_i) \oplus u(1).
$$
The unimodularity condition $\tr_\H X = 0$ translates as 
$$
\sum_i \alpha_i \tr (X_i) = 0
$$
where $\alpha_i$ are the multiplicities of the fundamental representations of $M_{N_i}(\F_i)$ appearing in $\H$. Using the above property for the traces on simple matrix Lie algebras, we find that unimodularity is equivalent to
$$
\sum_{l=1}^C \alpha_{i_l} z_{i_l} = 0 
$$
where the sum is over the complex factors (i.e. for which $\F_i = \C$) in $A$, labeled by $i_1, \ldots, i_C$.

We conclude that
$$
su(\A) \simeq \bigoplus_{i=1}^K su(N_i) \oplus u(1) ^{\oplus(C-1)}
$$
where $su(N_i)$ denotes $o(N_i), su(N_i)$ or $sp(N_i)$ depending on whether $\F_i = \R,\C$ or $\Qu$, respectively.\\

In order to get $su(\A) \simeq u(1) \oplus su(2) \oplus su(3)$, we either need that $C = 1$ (with the $u(1) \simeq so(2)$ coming from $M_2(\mathbb{R})$) or $C = 2$. In the first case $su(2)$ must come from from either $u(\Qu)$ or from $su(2) \simeq so(3)$, $u(M_3(\mathbb{R})) = o(3)$, i.e.
 \begin{itemize}
				\item $N_1 = 3_\C, N_2 = 1_\Qu, N_3 = 2_\R$ or 
				\item $N_1 = 3_\C, N_2 = 3_\R, N_3 = 2_\R$. 
\end{itemize}
In the second case we have the following options:
\begin{itemize}
				\item $N_1 = 3_\C, N_2 = 1_\Qu, N_3 = 1_\C$,
				\item $N_1 = 3_\C, N_2 = 2_\C$ or
				\item $N_1 = 3_\C, N_2 = 3_\R, N_3 = 1_\C$.
\end{itemize}
Modulo extra summands $\R$ these are the five options for the algebra $\A$.
\end{proof}

If the Hilbert space is to contain at least one copy of all the SM representations, then the algebra should allow for at least singlet, doublet and triplet representations. Only the third of these options satisfies this demand\footnote{Although the possible extra summands $\R$ do provide singlets too, the corresponding particles would lack any gauge interactions and are thus unsuitable.}.

\section{Constraints on the finite Hilbert space}

\subsection{Gauge group}\label{sec:gauge}

The gauge group of the Standard Model is known (see e.g.~\cite[\S 3.1]{BH10}) to be
\begin{align}
	\mathcal{G}_{SM} = U(1) \times SU(2) \times SU(3)/\mathbb{Z}_6,\label{eq:gsm}
\end{align}
where the finite abelian subgroup $\mathbb{Z}_6$ stems from the fact that certain elements of $U(1) \times SU(2) \times SU(3)$ act trivially on the Standard Model fermions. 

In the previous paragraph we have demanded that extensions of the SM Hilbert space still have a similar gauge group as that of the SM. Let us explicate this a bit. If we write $u = (u_1, \ldots, u_M)$ for a generic element of $U(\A)$, then the demand $\det_{\H}(u) = 1$ ---that is part of the definition of $SU(\A)$--- applied to the Hilbert space of \eqref{eq:Hilbertspace} translates to
	\begin{align}
		\det{}_\H(u) &= \prod_{i \leq j}^K\det{}_{N_i}(u_i)^{M_{N_iN_j}N_j}\det{}_{N_j}(u_j)^{M_{N_iN_j}N_i} = 1.\label{eq:det}
	\end{align}
	Applying this to the algebra of the SM ---$\com \oplus \mathbb{H} \oplus M_3(\com)$ and correspondingly $U(\A) = U(1) \times SU(2) \times U(3) \ni u = (\lambda, q, m) $--- we have five possibilities for the representations: $\mathbf{N}_i = \mathbf{1}, \overline{\mathbf{1}}, \mathbf{2}$, $\mathbf{3}$ and $\overline{\mathbf{3}}$ on which the determinant equals $\lambda, \lambda^{-1}, 1, \det m$ and $(\det m)^{-1}$ respectively. The relation \eqref{eq:det} then becomes
	\begin{align*}
		\det{}_\H(u) = \lambda^a \det(m)^b = 1
	\end{align*}
	where 
\begin{subequations}
\label{eq:ab}
\begin{align}
	a  &= 2M_{11} - 2M_{\bar1\bar1} + 2(M_{12}  - M_{\bar1 2}) + 3(M_{13} + M_{1\bar3} - M_{\bar13} - M_{\bar1\bar3})\label{eq:a},\\
	b  &= M_{31} + M_{3\bar1} - M_{\bar31} - M_{\bar3\bar1} + 2(M_{32} - M_{\bar32}) + 6(M_{33} - M_{\bar3\bar3}).\label{eq:b}
\end{align}
\end{subequations}
The multiplicity $M_{1\bar1}$ does not enter in the expression for $a$ above, for it actually appears twice but with opposite sign.

\begin{lem}
If $a$ divides $b$ then we have for the gauge group
$$
SU(\A) \simeq \left( U(1) \times SU(3)\right)/\Z_3 \times SU(2) \times \Z_a 
$$
\end{lem}
\proof
We show this in two steps: 
\begin{align}
\tag{I}
&SU(\A) \simeq  G \times SU(2) \times SU(3) / \Z_{3} \intertext{where $G =\left \{ (\lambda,\mu) \in U(1) \times U(1): \lambda^a \mu^{3b} = 1 \right\}$, containing $\Z_{3}$ as the subgroup $\{e \} \times \Z_{3}$, and}
\tag{II} 
&G \simeq \Z_a \times U(1).
\end{align}
For I, consider the following map
$$
(\lambda,\mu, q, m) \in G \times SU(2) \times SU(3) \mapsto (\lambda, q, \mu m) \in SU(\A).
$$
We claim that this map is surjective and has kernel $\Z_{3}$. 
If $(\lambda,q, m) \in SU(\A)$ then there exists a $\mu \in U(1)$ such that $\mu^3 = \det m \in U(1)$. Since $\lambda^a \mu^{3b} = \lambda^a \det m^b = 1$ the element $(\lambda,\mu, q, m)$ lies in the preimage of $(\lambda,q, m)$. 
The kernel of the above map consists of pairs $(\lambda,\mu,q, m) \in G \times SU(2) \times SU(3)$ such that $\lambda =1$, $q = 1$ and $m = \mu^{-1} 1_3$. Since $m \in SU(3)$ this $\mu$ satisfies $\mu^{3} = 1$. So we have established I.

For II we show that the following sequence is split-exact:
$$
1 \to U(1) \to G \to \Z_a \to 1,
$$ 
where the group homomorphisms are given by $\lambda \in U(1) \mapsto (\lambda^{3b/a}, \lambda^{-1}) \in G$ and $(\lambda,\mu) \in G \to \lambda \mu^{3b/a} \in \Z_a$. Exactness can be easily checked, and the splitting map is given by $\lambda \in \Z_a \to (\lambda, 1) \in G$. In this abelian case, the corresponding action of $\Z_a$ on $U(1)$ is trivial so that the resulting semidirect product $G \simeq U(1) \rtimes \Z_a \simeq U(1) \times \Z_a$.
\endproof

\begin{rmk}
In the case that we only allow for representations that already enter in the $\nu$SM (i.e.~$\mathbf{1}, \overline{\mathbf{1}}, \mathbf{2}$ and $\mathbf{3}$), $a$ and $b$ are given by:
\begin{subequations}
\label{eq:ab2}
\begin{align}
	a  &= 2M_{11} - 2M_{\bar1\bar1} + 2(M_{12} - M_{\bar1 2}) + 3(M_{13} - M_{\bar13}) \label{eq:a2},\\
	b  &= M_{31} + M_{3\bar1} + 2M_{32} + 6M_{33}.\label{eq:b2}
\end{align}
\end{subequations}
The Standard Model itself (see also \S\ref{sec:sm}) is given by $M_{11} = M_{1\bar 1} = M_{12} = M_{13}=M_{\bar 1 3} =M_{23} = 3$ (for three families) and all other multiplicities zero. In this case, $a=b=12$ so that the above Lemma yields $SU(\A) \simeq SU(2) \times (U(1)  \times SU(3))/ {\Z_{3}} \times \Z_{12}$ in concordance with what was found in \cite[Prop. 2.16]{CCM}. The representation of $SU(\A)$ on $\H$ is as $u \mapsto u JuJ^{-1}$, and the kernel of this representation is $\Z_2$ \cite[Prop. 6.3]{SD11}. In turn, we find that $SU(\A)/\Z_{2} \simeq \mathcal{G}_{SM} \times \Z_{12}$.
\end{rmk}

Note that from the definition of $SU(\A)$ we can determine the hypercharges of the particles from the $U(1)$ factor of $SU(\A)$ (cf. \cite[Prop. 2.16]{CCM}):
\begin{align}
	\{(\lambda, 1, \lambda^{-a/3b}1_3), \lambda \in \C, |\lambda| = 1\} \subset SU(\A)\label{eq:u1charge},
\end{align}
where the value of $a/b$ must be determined via additional constraints.\footnote{Note however, that these hypercharges are determined up to an overall power.} The corresponding {\it hypercharge generator} $Y = i (1,0,-a/3b  1_3)$ then acts on $\H_F$ as $(Y \otimes 1 - 1 \otimes Y^o)$. The hypercharges for the example of the Standard Model then come out right, as illustrated in Table \ref{tab:hypercharges-SM}.

\begin{table}
\begin{tabularx}{\textwidth}{ccrl}
\toprule
Representation  & Multiplicity & Hypercharge & SM-particles\\
\midrule
${\bf 1} \otimes {\bf 1}^o$ & $3$ & $1-1=0$& right-handed neutrino\\
${\bf \bar 1} \otimes {\bf 1}^o$ & $3$ & $-1-1=-2$& right-handed electron\\
${\bf 2} \otimes {\bf 1}^o$ & $3$ & $-1$& left-handed leptons\\
${\bf 1} \otimes {\bf 3}^o$ & $3$ & $1+1/3=4/3$& right up-type quarks\\
${\bf \bar 1} \otimes {\bf 3}^o$ & $3$ & $-1+1/3=-2/3$& right down-type quarks\\
${\bf 2} \otimes {\bf 3}^o$ & $3$ & $1/3$ & left quarks\\
\bottomrule
\end{tabularx}
\caption{Hypercharges as derived from the finite spectral triple describing the Standard Model.}
\label{tab:hypercharges-SM}
\end{table}

\subsection{Anomalies and anomaly cancellation}\label{sec:anomalies}
 
In short, a \emph{quantum anomaly} is said to arise when a certain local (gauge) symmetry of a classical theory gets broken upon quantization. Here, we focus attention on the \emph{non-abelian chiral gauge anomaly}: even if the action is invariant under the transformation $\H \ni \zeta \to \exp(\gamma^5 T)\zeta$ (with $T$ anti-Hermitian), the path integral corresponding to this action might not be. The demand of having an anomaly free theory (i.e.~a theory that can be quantised) can be casted in an expression that depends on the fermionic content of the theory. We will use it to put constraints on the particle content described by the spectral triple.\\ 

In what follows, we assume a finite Hilbert space of the form $\H_F = \H_L \oplus \H_R \oplus \H_{L}^o \oplus \H_R^o$ (containing the left- and right-handed particles and antiparticles respectively) and KO-dimension $6$ (i.e.~$\gamma_F$ anti-commutes with $J_F$) implying that a generic element $\zeta \in \H^+$ ---the physical Hilbert space--- is of the form 
\begin{align}
	\zeta = \xi_L \otimes e_L + \xi_R \otimes e_R + \eta_R \otimes \bar e_L + \eta_L \otimes \bar e_R.\label{eq:spinor_decomp}
\end{align}
The motivation for working in this KO-dimension $6$ is that we are after theories beyond the Standard Model and this is the KO-dimension in which the latter is described. Furthermore we take the finite Dirac operator $D_F$ to be zero.\footnote{This corresponds to having massless fermions, which is crucial when considering the chiral gauge anomaly.} $T^a$ will denote a \emph{fixed} generator of the gauge Lie algebra $su(\A_F)$, i.e.~$(T^a)^* = - T^a$.
	
\begin{lem}
	Let $\zeta \in \H^+$ and let $\alpha : M \to \mathbb{R}$ be a real function, then the non-abelian chiral gauge transformation
	\begin{align}
	\zeta &\to  U_J\zeta,\text{ with } U_J := \exp(\alpha  \gamma \mathbb{T}^a), \mathbb{T}^a = \pi(T^a) - J\pi(T^a)^* J^*\label{eq:anomtrans}
	\end{align}
is an on-shell symmetry of the fermionic action in \eqref{eq:totalaction}.
\end{lem}
\begin{proof}
We only have to consider the fermionic part of \eqref{eq:totalaction}, since the gauge fields do not transform. For $U_J$ as in \eqref{eq:anomtrans} with $ \gamma = \gamma^5  \otimes \gamma_F$ one can easily prove that 
\begin{align*}
	J U_J  &= U_J^* J &&\text{and}& U_J\gamma^\mu &= \gamma^\mu U_J^*,
\end{align*} 
by using that left and right actions of the algebra on the Hilbert space must commute and that $\{\gamma^5, \gamma^\mu\} = \{J, \gamma_F\} = 0$. Then for the inner product we have, upon transforming $\zeta$:
\begin{align*}
	\frac{1}{2}\langle JU_J\zeta, D_AU_J\zeta\rangle &= \frac{1}{2}\langle J\zeta, i\gamma^\mu U_J^*(\nabla^S_\mu + \mathbb{A}_\mu)U_J\zeta\rangle \\
&= \frac{1}{2}\langle J\zeta, \slashed{\partial}\zeta\rangle + \frac{1}{2}\langle J\zeta, i \partial_\mu(\alpha) \gamma^\mu \gamma\,\mathbb{T}^a\zeta\rangle  + \frac{1}{2}\langle J\zeta, i\gamma^\mu \exp[-\alpha \gamma\, \ad\mathbb{T}^a]\mathbb{A}_\mu\zeta\rangle, 
\end{align*}
where we have used the identity $\exp(A)B\exp(- A) = \exp(\ad A)B$ for complex $n \times n$ matrices $A$ and $B$. If we expand the exponential in the last term of the previous expression, and use partial integration for the second term, this becomes
\begin{align*}
 & \frac{1}{2}\langle J\zeta, D_A\zeta\rangle - \frac{1}{2}\int_M \alpha \partial_\mu \big[(J\zeta, i \gamma^\mu \gamma \mathbb{T}^a\zeta)\sqrt{g}\big]\mathrm{d}^4x - \frac{1}{2}\langle J\zeta, i\alpha \gamma^\mu\gamma [\mathbb{T}^a,\mathbb{A}_\mu]\zeta\rangle + \mathcal{O}(\alpha^2).
\end{align*}
Writing $\mathbb{A}_\mu = {A}_{\mu\,b} \mathbb{T}^b$ and using $[\mathbb{T}^a, \mathbb{T}^b] = \ad[T^a, T^b] = f^{ab}_{\phantom{ab}c}\mathbb{T}^c$ this implies
\begin{align*}
 \frac{1}{2} \langle JU_J\zeta, D_AU_J\zeta\rangle &= \frac{1}{2}\langle J\zeta, D_A\zeta\rangle - \frac{1}{2}\int_M \alpha (D_{\mu\,c}^{a} j^{\mu\,c})\mathrm{d}^4x + \mathcal{O}(\alpha^2)
\end{align*}
where 
\begin{align*}
D_{\mu\,c}^{a} &= \partial_\mu \delta^{a}_{\phantom{a}c} + f^{ab}_{\phantom{ab}c}A_{\mu\, b},& j^{\mu\, a} &= (J\zeta, i\gamma^\mu\gamma \mathbb{T}^a\zeta)\sqrt{g}.
\end{align*}
On the other hand we have the following:
\begin{align*}
		\partial_\mu j^{\mu\,a} = (J\nabla^S_\mu\zeta, i\gamma^\mu\gamma \mathbb{T}^a\zeta)\sqrt{g} + (J\zeta, i\nabla^S_\mu\gamma^\mu\gamma \mathbb{T}^a\zeta)\sqrt{g}+ (J\zeta, i\gamma^\mu\gamma \mathbb{T}^a\zeta)\partial_\mu\sqrt{g},
\end{align*}
where we have used that the spin-connection is Hermitian and commutes with $J$. Using that for the tensor density $\sqrt{g}$ we have
$
	\partial_\mu \sqrt{g} = \Gamma^{\lambda}_{\lambda\mu}\sqrt{g}
$ 
(e.g.~\cite[\S 21.2]{misner}) and that for the spin-connection $[\nabla^S_\mu, \gamma^\mu] = -\Gamma^\mu_{\mu\lambda}\gamma^\lambda$ this yields
\begin{align*}
\partial_\mu j^{\mu\,a} = (J\nabla^S_\mu\zeta, i\gamma^\mu\gamma \mathbb{T}^a\zeta)\sqrt{g} + (J\zeta, i\gamma^\mu\gamma \mathbb{T}^a\nabla^S_\mu\zeta)\sqrt{g}.
\end{align*}
Employing the Dirac equation for $\zeta$ and the skew-adjointness of $\mathbb{A}_\mu$ this gives
\begin{align*}
\partial_\mu j^{\mu\,a} = - (J\zeta, i\gamma^\mu\gamma [\mathbb{T}^a,\mathbb{A}_\mu]\zeta)\sqrt{g}.
\end{align*}
i.e.
$
	D_{\mu\,c}^{a} j^{\mu\, c} = 0,
$ 
establishing the result.
\end{proof} 

In order to progress, we need to be a bit more detailed. For a representation $\mathbf{N_i} \otimes \mathbf{N_j}^o$ of a given pair $(i, j)$, $i \leq j$, the representation $\pi(T^a)$ can be decomposed as $\pi_L(T^a) + \pi_R(T^a)$, where one of the two is trivial depending on the chirality of $\mathbf{N_i} \otimes \mathbf{N_j}^o$. The representation on the conjugate of $\mathbf{N_i} \otimes \mathbf{N_j}^o$ is denoted by $\bar\pi_L(T^a) + \bar\pi_R(T^a)$. Hence we write for the full representation $\pi(T^a)$ on $\mathbf{N_i} \otimes \mathbf{N_j}^o \oplus \mathbf{N_i} \otimes \mathbf{N_j}^o$:
\begin{align}
\pi(T^a) = \pi_L(T^a) + \pi_R(T^a) +\bar\pi_L(T^a) + \bar\pi_R(T^a).\label{eq:rep_decomp}
\end{align}
\begin{exmpl}
	For the case $\H_F = \mathbf{2}_L \otimes \mathbf{3}^o \oplus \mathbf{3} \otimes \mathbf{2}^o_L$ and $T^a = \tau^a$, the Gell-Mann matrices, we have
	\begin{align*}
		\pi_R(\tau^a) &= \bar\pi_R(\tau^a) = \pi_L(\tau^a) = 0,& \bar\pi_L(\tau^a) &= \tau^a.
	\end{align*}
	For the case $\H_F = \mathbf{1}_R \otimes \mathbf{3}^o \oplus \mathbf{3} \otimes \mathbf{1}^o_R$ and $T^a \equiv Y = i(1, 0, -\tfrac{a}{3b}1_3)$, the hypercharge generator, we have
	\begin{align*}
		\pi_R(Y) &= i, & \bar\pi_R(Y) &= -i\frac{a}{3b}1_3, & \pi_L(\lambda) &= \bar\pi_L(\lambda) = 0.
	\end{align*}
\end{exmpl}
Since $\langle J \xi,D_A \xi \rangle$ defines an anti-symmetric bilinear form \cite[Prop. 4.1]{CCM} we can also write the chiral gauge transformation in terms of the component spinors \eqref{eq:spinor_decomp}, reading:
\begin{align*}
	\frac{1}{2}\langle JU_J\zeta, D_A U_J\zeta\rangle = \langle J_M(U^o)^*\eta, D_AU\xi\rangle 
\end{align*}
i.e.~under the transformation \eqref{eq:anomtrans} we have 
\begin{align}
	\xi &\mapsto \exp(\alpha \gamma\, \diag\{\pi_L - \bar\pi^o_L, \pi_R - \bar \pi_R^o\}(T^a)) \xi \equiv U\xi, & \eta &\mapsto (U^*)^o\eta \label{eq:anomtrans_spinors}
\end{align}
in the notation of \eqref{eq:rep_decomp}.\\ 

Now considering the path integral
\begin{align}
	\int \mathcal{D}\eta\mathcal{D}\xi \mathcal{D}\mathbb{A}_\mu \exp(S[\eta,\xi, \mathbb{A}_\mu]),\label{eq:path_integral}
\end{align}
(where in a Euclidean set up the fields $\xi$ and $\eta$ should be considered as independent) there is a second effect from transforming the fermionic fields \eqref{eq:anomtrans_spinors}, which is from the transformation of its measure. The following derivation is primarily based on \cite[\S 22.2]{W05}, \cite[\S 5.6]{BERT96} and \cite[\S 13.2]{NAK90}. We first consider the effect of that transformation on the Dirac spinor $\xi = (\xi_L, \xi_R)$ with finite component. Let 
\[
\xi = \sum_I a_I\psi_I
\]
be its decomposition into the eigenfunctions of $D$. Here $I$ is a generic index describing both continuous and finite indices: if, for example, the particle sector of $\H_F$ would equal $(\mathbf{N} \otimes \mathbf{M}^o)_L + (\mathbf{K} \otimes \mathbf{L}^o)_R$ then
\begin{align*}
 	\sum_I a_I\psi_I \equiv \sum_{i}\Big(\sum_{n,m} a_{inm} \psi_{iL} \otimes e_n \otimes \bar e_m + \sum_{k,l} a_{ikl} \psi_{iR} \otimes e_k \otimes \bar e_l\Big),\quad a_{inm} = \langle \psi_{iL} \otimes e_n \otimes \bar e_m, \xi\rangle,
\end{align*}
where $\psi_{iL,R}$ is a left-/right-handed eigenspinor of $\slashed{\partial}$ and $\{e_n \otimes \bar e_m, 1 \leq n \leq N, 1\leq m \leq M\}$ denotes the basis of the left-handed finite part. Then the transformation
\begin{align*}
	\xi \mapsto \xi' :=  U\xi 
\end{align*}
sets
$
	\xi' = \sum_{I} a_{I}'\psi_I $ with $a_{I}' = \sum_{J} C_{IJ}a_{J}
$
where 
\[ C_{IJ} = \delta_{IJ} + \langle \psi_I, \alpha \gamma \diag\{ (\pi_L  - \bar\pi_L^o)(T^a), (\pi_R - \bar\pi_R^o)(T^a)\}\psi_J\rangle + \mathcal{O}(\alpha^2). \]
The effect of the transformation on the measure $\mathcal{D}\xi$ in \eqref{eq:path_integral} is then
\begin{align*}
\mathcal{D}\xi \to \det(C)^{-1}\mathcal{D}\xi,
\end{align*}
and similarly for $\eta$. Writing $X := \diag\{ (\pi_L  - \bar\pi_L^o)(T^a), (\pi_R - \bar\pi_R^o)(T^a)\}$, and calculating the determinant by using $\det A = \exp\tr \ln A$ gives
\begin{align*}
	\det(C_{IJ}) &= \exp\tr \ln (\delta_{IJ} + \langle \psi_I , \alpha \gamma X\psi_J \rangle + \mathcal{O}(\alpha^2)) \\
				&\approx \exp\tr \langle \psi_I , \alpha \gamma X\psi_J\rangle\\
				&= \exp \sum_{I}\langle \psi_I , \alpha \gamma X\psi_I \rangle,
\end{align*}
where in going from the first to the second line we have used that $\ln(1 + z) = z + \mathcal{O}(z^2)$, and that $\alpha$ is infinitesimal. The anomaly corresponding to the transformation of $\xi$ is thus
\begin{align*}
	\exp(- \mathscr{A}_\xi),\qquad \mathscr{A}_\xi = \int_M  \alpha\,\mathscr{A}_\xi(x)\,\sqrt{g}\mathrm{d}^4x,\qquad \mathscr{A}_\xi(x) = \sum_{I}(\psi_I, \gamma\,X\psi_I).
\end{align*}
Even though this is an ill-defined quantity, we can make sense of it using the following regularization scheme: 
\begin{align*}
\sum_{I}\langle \psi_I, \alpha\gamma\,X\psi_I\rangle &:= \lim_{\Lambda \to \infty} \sum_{I}\langle \psi_I, \alpha\gamma\,X h(\lambda_{I}^2/\Lambda^2)\psi_I\rangle
\end{align*}
where $h$ can be any function that satisfies $h(0) = 1$ and $\lim_{x \to \infty} h(x) = 0$, and $\lambda_{I}$ is the eigenvalue of $D_A$ with eigenvector $\psi_I$:
\begin{align*}
						\mathscr{A}_\xi = \lim_{\Lambda \to \infty} \sum_{I}\langle \psi_I, \alpha\gamma\,X h(D_A^2/\Lambda^2)\psi_I\rangle
						= \lim_{\Lambda \to \infty} \tr_{\H_f}[\gamma\,X h(D_A^2/\Lambda^2)], 
\end{align*}
where with $\H_f = L^2(M, S)_L \otimes \H_L \oplus L^2(M, S)_R \otimes \H_R$ we mean the \emph{particle sector} of the total Hilbert space $\H$, as opposed to the \emph{anti-particle sector} $\H_{\bar f} = L^2(M, S)_R \otimes \H_L^o \oplus L^2(M, S)_L \otimes \H_R^o$. Using a heat kernel expansion \eqref{eq:total_action_exp} this expression equals
\begin{align}
\mathscr{A}_\xi &= \lim_{\Lambda \to \infty} \tr_{\H_f} \alpha\gamma\,X\Big[h_4\Lambda^2e_0(x, D_A^2) + h_2\Lambda^2 e_2(x, D_A^2) + h(0)e_4(x, D_A^2) + \mathcal{O}(\Lambda^{-2}) \Big]\label{eq:anom-xi}.
\end{align}
In an analogous fashion we can determine the anomaly that is caused by the spinor $\eta$, yielding
\begin{align}
  \mathscr{A}_{\eta} &= - \lim_{\Lambda \to \infty}  \tr_{\H_{\bar f}} \alpha\gamma\,X^o\Big[h_4\Lambda^2e_0(x, D_A^2) + h_2\Lambda^2 e_2(x, D_A^2) + h(0)e_4(x, D_A^2) + \mathcal{O}(\Lambda^{-2}) \Big]\label{eq:anom-eta}.
\end{align}
Then we have the following result:
\begin{lem}
	The total anomaly in the path integral due to \eqref{eq:anomtrans} is given by:
	\begin{align*}
	 \mathscr{A} = \int_M \alpha \mathscr{A}(x) \sqrt{g}\mathrm{d}^4x
	\end{align*}
	with 
	\begin{align}
\mathscr{A}(x) &=	\frac{1}{16\pi^2}\epsilon^{\mu\nu\lambda\sigma}F_{\mu\nu}^bF_{\lambda\sigma}^c\Big[ \tr_{\H_L}  (\pi_L - \bar\pi_L^o)(T^a) \{\mathbb{T}^b, \mathbb{T}^c\} -\tr_{\H_R}  (\pi_R - \bar\pi_R^o)(T^a)\{\mathbb{T}^b, \mathbb{T}^c\}\Big]\nonumber\\
	&\qquad + \frac{1}{192\pi^2}\epsilon^{abcd}R^{\mu\nu}_{\phantom{\mu\nu}ab}R_{\mu\nu\,cd}\Big[ \tr_{\H_L}  (\pi_L - \bar\pi_L^o)(T^a)  -\tr_{\H_R}  (\pi_R - \bar\pi_R^o)(T^a)\Big]\label{eq:anomresult}.
	\end{align}
\end{lem}
\begin{proof}
We start with \eqref{eq:anom-xi}, taking the expression for $e_{0,2,4}(x, D_A^2)$ from \eqref{eq:gilkey3}, 
where, for an almost-commutative geometry, $E$ and $\Omega_{\mu\nu}$ are determined by the field strength of the gauge field and the Riemann tensor for $M$:
\begin{align*}
	E &= \frac{1}{4} R  + \frac{1}{2}\gamma^\mu\gamma^\nu \mathbb{F}_{\mu\nu},& \Omega_{\mu\nu} = \frac{1}{4}R^{ab}_{\mu\nu}\gamma_{ab}\otimes 1 + \mathbb{F}_{\mu\nu}, \end{align*}
with $\gamma_{ab} = \tfrac{1}{2}[\gamma_a, \gamma_b]$. Since $\tr \gamma^5 = \tr\gamma^5\gamma^\mu\gamma^\nu = 0$ we only retain, after taking the limit in $\Lambda$, 
\begin{align}
\mathscr{A}_\xi&= h(0)\tr_{\H_f} \alpha \gamma\,X e_4(x, D_A^2) \nonumber\\ 
&= \frac{h(0)}{(4\pi)^2}\frac{1}{360}\tr_{\H_f} \alpha \gamma\,X \big[180E^2 + 30\Omega_{\mu\nu}\Omega^{\mu\nu}\big], 
\end{align}
plus boundary terms, that vanish upon integration since the manifold $M$ is without boundary. Inserting the expressions for $E$ and $\Omega^{\mu\nu}$, performing the trace over Dirac indices and setting $h(0) = 1$ this becomes
\begin{align*}
	\mathscr{A}_\xi &= \int_M \alpha \tr_{\H_L \oplus \H_R}\bigg[\frac{1}{32\pi^2}\epsilon^{\mu\nu\lambda\sigma}\gamma_F X \mathbb{F}_{\mu\nu}\mathbb{F}_{\lambda\sigma} + \frac{1}{384\pi^2}\epsilon^{abcd}R^{\mu\nu}_{\phantom{\mu\nu}ab}R_{\mu\nu\,cd}\gamma_F X\bigg]\sqrt{g}\mathrm{d^4}x.
\end{align*}
The derivation of $\mathscr{A}_\eta$ can be found using the same arguments, reading 
\begin{align*}
	\mathscr{A}_\eta &= - \int_M \alpha\tr_{\H_L^o \oplus \H_R^o}\bigg[ \frac{1}{32\pi^2}\epsilon^{\mu\nu\lambda\sigma}  \gamma_F X^o \mathbb{F}_{\mu\nu}\mathbb{F}_{\lambda\sigma} + \frac{1}{384\pi^2}\epsilon^{abcd}R^{\mu\nu}_{\phantom{\mu\nu}ab}R_{\mu\nu\,cd} \gamma_F X^o\bigg]\sqrt{g}\mathrm{d}^4x.
\end{align*}
Adding $\mathscr{A}_\xi$ and $\mathscr{A}_\eta$ and putting in the expression for $X$, using that $\{\gamma_F, J_F\} = 0$ and $\tr_{\H^o}X^o = \tr_{\H}X$ (so contributions from the particle and antiparticle sectors add), the total anomaly reads
\begin{align*}
	\mathscr{A} &= \int_M \alpha\mathscr{A}(x)\sqrt{g} \mathrm{d}^4x\\
	&\text{with}\quad \mathscr{A}(x) =\frac{1}{16\pi^2}\epsilon^{\mu\nu\lambda\sigma}\Big[ \tr_{\H_L}  (\pi_L - \bar\pi_L^o)(T^a) \mathbb{F}_{\mu\nu}\mathbb{F}_{\lambda\sigma} -\tr_{\H_R}  (\pi_R - \bar\pi_R^o)(T^a) \mathbb{F}_{\mu\nu}\mathbb{F}_{\lambda\sigma}\Big]\\
	&\qquad\qquad\qquad + \frac{1}{192\pi^2}\epsilon^{abcd}R^{\mu\nu}_{\phantom{\mu\nu}ab}R_{\mu\nu\,cd}\Big[ \tr_{\H_L}  (\pi_L - \bar\pi_L^o)(T^a)  -\tr_{\H_R}  (\pi_R - \bar\pi_R^o)(T^a)\Big],
\end{align*}
where we have used that $\gamma_F\big|_{\H_L} = - \gamma_F\big|_{\H_R} = 1$. Writing $\mathbb{F}_{\mu\nu} = F_{\mu\nu}^a \mathbb{T}^a$ and exploiting the (anti)symmetries of $\epsilon^{\mu\nu\lambda\sigma}$ and the field strength tensor $F_{\mu\nu}$ we obtain \eqref{eq:anomresult}.
\end{proof}

This result should hold for any generator $T^a$ of the Lie algebra $su(\A)$ of the gauge group (cf.~\eqref{eq:lie-alg}). So, if we want the theory to be anomaly free, \eqref{eq:anomresult} should be zero. \\

We apply this general result to the models for which the Lie algebra is $su(\A) = u(1) \oplus su(2) \oplus su(3)$; the $T^a$ can separately be $Y \equiv i(1, 0, -\frac{a}{3b}1_3)$, $(0, \sigma^a, 0)$ and $(0, 0, \tau^a)$, the generators of the Lie algebra of the SM. Here $\sigma^a$ and $\tau^a$ denote the Pauli and Gell-Mann matrices respectively. Since $\tr\sigma^i = \tr \tau^a = 0$, the gravitational term of \eqref{eq:anomresult} only gets a contribution from the $u(1)$-part: 
\begin{align*}
	\tr_{\H_L} (Y \otimes 1 - 1 \otimes Y^o) - \tr_{\H_R} (Y \otimes 1 - 1\otimes Y^o),\quad Y = i(1, 0, -a/3b\,1_3),
\end{align*}
where although $Y$ denotes the hypercharge \emph{generator} of $u(1)$, it is $Y \otimes 1 -1\otimes Y^o$ that represents the actual hypercharge. For the representations that appear in the Standard Model, this gives
\begin{align}
	&2(-1)M_{21} + 2\cdot 3\frac{a}{3b} M_{23} - 3\bigg(1 + \frac{a}{3b}\bigg)M_{13} - (-2) M_{\bar11} - 3\bigg(-1 + \frac{a}{3b}\bigg)M_{\bar13}\nonumber\\
	&= - 2M_{21} + 2\frac{a}{b} M_{23} - \bigg(3 + \frac{a}{b}\bigg)M_{13} + 2 M_{\bar11} + \bigg(3 - \frac{a}{b}\bigg)M_{\bar13} = 0\label{eq:SManom2}.
\end{align}
Now for the non-gravitational term in \eqref{eq:anomresult}. If we use the cyclicity of the trace, $\tr(\sigma^a) = \tr(\tau^a) = 0$ and $\tr\sigma^a\sigma^b = \tr\tau^a\tau^b = 2\delta^{ab}$, $\{\sigma^a,\sigma^b\} = 2\delta^{ab}$ and $\{\tau^a,\tau^b\} = \tfrac{4}{3}\delta^{ab} + 2d^{ab}_{\phantom{ab}c} \tau^c$, we find ---when restricting to the representations that already appear in the Standard Model--- it to be equivalent to the following relations:
\begin{subequations}
\label{eq:SManom}
\begin{align}
	\intertext{$T^a$ the hypercharge generator, $F_{\mu\nu}$ the photon field:}
	6\bigg(\frac{a}{3b}\bigg)^3M_{23} - 3\bigg(1 + \frac{a}{3b}\bigg)^3M_{13} - 3\bigg(-1 + \frac{a}{3b}\bigg)^3M_{\bar 13} + 2(-1)^3M_{21} - (-2)^3M_{\bar11} &= 0\label{eq:SManoma},\\
	\intertext{$T^a$ the hypercharge generator, $F_{\mu\nu}$ the gluon field:}
	2\frac{a}{b}M_{23} - 3\bigg(1 + \frac{a}{3b}\bigg)M_{13} - 3\bigg(- 1 + \frac{a}{3b}\bigg)M_{\bar 13} &= 0\label{eq:SManomb},\\
	\intertext{$T^a$ the hypercharge generator, $F_{\mu\nu}$ the $SU(2)$ boson field:}
	-2M_{21} + 2\frac{a}{b}M_{23} &= 0\label{eq:SManomc},\\
	\intertext{$T^a$ a Gell-Mann matrix, $F_{\mu\nu}$ the gluon field:}
	2M_{23} - M_{13} - M_{\bar13} &= 0.\label{eq:SManomd}
\end{align}
\end{subequations}
All other combinations are seen to vanish. It is evident that the demand of the cancellation of anomalies puts rather stringent constraints on the multiplicities. One well-known result can now be re-confirmed.

\begin{prop}
The minimal Standard Model is anomaly free. 
\end{prop}
\proof
A glance back at the remark just below \eqref{eq:ab2} shows that the non-zero multiplicities of the representation of $\A_F =\C \oplus \Qu \oplus M_3(\C)$ in $\H_F$ are $M_{11} = M_{\bar 11} = M_{21} = M_{13}=M_{\bar 1 3} = M_{23} = 3$ and $a=b=12$, for which the above equations are readily seen to hold. 
\endproof

Two comments are in order here. In using the demand of anomaly cancellation for validating extensions of the Standard Model, we need to know what the value of the grading (chirality) $\gamma$ is on each of the representations. However, there is no such thing as a canonical expression for a grading\footnote{A rather strict demand on the grading would be that of \emph{orientability} \cite{C96, KR97}: for a set $I\ni i$ there exist $a_i, b_i \in \A$ such that $\gamma = \sum_i a_i \otimes b_i^o$. But since already the grading of the SM (with a right-handed neutrino included) is not of this form \cite{ST06}, demanding it for $\gamma$ would in fact be too strict.}, which in principle limits the scope of these constraints to the particles we know the chirality of, i.e.~the SM particles. For non-SM representations we can only try both possible gradings separately. Note, secondly, that the demand for a theory to be anomaly free is most often used for determining the hypercharges of the particles involved. Here, however, these are already determined by the constraint concerning the gauge group, causing the role of anomaly cancellation to be different; it may be used to put constraints on the multiplicities of the representations.

\subsection{GUT relation}\label{sec:gut}

In determining the bosonic part of the action functional \eqref{eq:totalaction} by means of a heat kernel expansion, one obtains \cite{CM97} from the $a_4(x, D_A^2)$-term a contribution 
\begin{align*}
	- \frac{f(0)}{24\pi^2}\int_M \tr_{\H_F} \mathbb{F}_{\mu\nu}\mathbb{F}^{\mu\nu}.
\end{align*}
Here the trace runs over the entire (finite) Hilbert space $\H_F$, and $\mathbb{F}_{\mu\nu}$ is a generic symbol denoting the field strength associated to the various gauge fields. Calculating it explicitly ---assuming a Hilbert space of the form \eqref{eq:Hilbertspace}--- gives 
\begin{align*}
-	\frac{f(0)}{24\pi^2}\tr_{\H_F} \mathbb{F}_{\mu\nu}\mathbb{F}^{\mu\nu} &= -\frac{f(0)}{24\pi^2}\sum_{k} G_k(g_k)\tr_{N_k} \mathbb{F}_{\mu\nu}^{(k)}\mathbb{F}^{\mu\nu(k)}
\end{align*}
with the coefficients $G_k \equiv G_k(g_k)$ (as a function of the coupling constant $g_k$) being given by
\begin{align}
 	G_k  = 2\sum_{i, j \leq i} M_{N_i N_j}c^{(k)}_{ij}\big(q_{ij}^{(k)}\big)^2g_k^2.\label{eq:gut1}
\end{align}
Here the label $(k)$ denotes the \emph{type} of gauge field and $q^{(k)}_{ij} = -q^{(k)}_{ji}$ is the charge of the representation $\mathbf{N}_i \otimes \mathbf{N}_j^o$ associated to the gauge field $A^{(k)}_\mu$. For the number $c^{(k)}_{ij}$ we have in the case that the only representations are $\mathbf{1}$, $\overline{\mathbf{1}}$, $\mathbf{2}$ and $\mathbf{3}$:
\begin{align*}
	c^{(k)}_{ij} &= \begin{cases} N_i & \text{if\ } k = j,\\ N_j & \text{if } k= i, \\ N_iN_j & \text{else}. \end{cases}
\end{align*}
In the description of the SM from NCG, the three coefficients are precisely such that upon equating them to the normalisation constant $-1/4$ in front of the kinetic term ---as is customary--- they automatically satisfy the GUT-relation \eqref{eq:gutrel} \cite[\S 16.1]{CCM}. Now certainly in reality, with the particle content of the Standard Model alone, the coupling constants do not meet at a single energy scale. But first of all this feature is too specific to disregard it as a mere coincidence and secondly the entire predictive power of NCG relies \cite[\S 8]{SD11} on it: if it has to say anything more about reality than does the conventional approach to the Standard Model (or any of its extensions), we should take this feature seriously. Furthermore much of the `beyond the Standard Model' research has been conducted in a setting that is characterized by coupling constant unification. \emph{To this end we promote that the coupling constants satisfy the GUT relation from a feature to a demand.}\footnote{Certainly at some point one should check that a Hilbert space that satisfies the GUT-relation is compatible with a crossing of the coupling constants as obtained using the Renormalization Group Equations and the very same particle content.} The nature of the constants $G_k$ is then such that it allows one to put constraints on the Hilbert space.\\

\begin{table}[h]
\begin{tabularx}{\textwidth}{XccccXccccXccccX}
\toprule
 & 	 		& $G_1$ 		&$G_2$&$G_3$&& 				& $G_1$ 		&$G_2$&$G_3$&\\
\midrule
 & \rep{1}{1}		& 0 			& 0 & 0 & & \rep{\bar 1}{\bar 1}	& 0 			& 0 & 0 & \\ 
 & \rep{1}{\bar 1}	& 8 			& 0 & 0 & & \rep{\bar 1}{2}		& 4 			& 2 & 0 & \\ 
 & \rep{1}{2}		& 4 			& 2 & 0 & & \rep{\bar 1}{3}		&$6(-1+\frac{a}{3b})^2$	& 0 & 2 & \\ 
 & \rep{1}{3}		&$6(1+\frac{a}{3b})^2$	& 0 & 2 & & \rep{\bar 1}{\bar 3}	&$3(-1+\frac{a}{3b})^2$ & 0 & 2 & \\
 & \rep{1}{\bar 3}	&$6(1-\frac{a}{2b})^2$	& 0 & 2 & & 				&   			&   &   & \\
 &			&			&   &   & & \rep{3}{3} 			& 0  			& 0 & 12 & \\
 & \rep{2}{2}		& 0          & 8  & 0 & & \rep{3}{\bar3}		& $9(-\frac{2a}{3b})^2$ & 0 & 12& \\
 & \rep{2}{3}		&12$(\frac{a}{3b})^2$	& 6 & 4 & & \rep{\bar 3}{\bar3}		& 0			& 0 & 12 & \\
 & \rep{2}{\bar 3}	&12$(-\frac{a}{3b})^2$	& 6 & 4 & &				&   			&   & 	&\\
\bottomrule
\end{tabularx}
\caption{The contributions $c^{(k)}_{ij}q_{ij}^{(k)}$ for any of the possible representations of the Standard Model algebra contributes to either of the three coupling constants. These already include a factor $2$ from the fact that in taking the sum every representation appears twice.}
\label{tab:coeff}
\end{table}

In the case of the Standard Model algebra  there are three different gauge fields, since the gauge field acting on $\overline{\mathbf{1}}$ is the same as the one acting on $\mathbf{1}$. The $SU(2)$ charge equals $1$ on the representation $\mathbf{2}$ and $0$ on everything else and similarly the $SU(3)$ charge equals $1$ on the representation $\mathbf{3}$ and $0$ on everything else. The charges for the $U(1)$ gauge field are determined by \eqref{eq:u1charge}. The coefficients with which a certain representation contributes to any of the three coupling constants can be found in Table \ref{tab:coeff}. If we only allow the representations that already appear in the Standard Model ---$\mathbf{1}, \overline{\mathbf{1}}, \mathbf{2}, \mathbf{3}$--- we get for the three different coefficients in \eqref{eq:gut1}:
\begin{align*}
	 G_{1} &= 2\Big[4M_{1\bar 1} + 2M_{12} + 2M_{\bar 12} + 3\Big(1 + \frac{a}{3b}\Big)^2M_{13} + 3\Big(1 - \frac{a}{3b}\Big)^2M_{\bar 1 3} + 6\Big(\frac{a}{3b}\Big)^2M_{32}\Big],\\
	 G_{2} &= 2[M_{21} + M_{2\bar 1} + 3M_{23} + 4M_{22}],\\
	 G_{3} &= 2[M_{31} + M_{3\bar 1} + 2M_{32} + 6M_{33}].
\end{align*}
Now upon taking the trace of kinetic terms of the gauge fields, the second and third line get an additional factor $2$ because of the normalisation of the Pauli and Gell-Mann matrices respectively. The demand for the GUT-relation \eqref{eq:gutrel} is then that
\begin{align*}
	3G_1 = 10G_2 = 10G_3
\end{align*}
i.e. in terms of multiplicities
\begin{align}
&\phantom{= 10}\Big[12M_{1\bar1} + 6M_{12} + 6M_{\bar 12} + 9\Big(1 + \frac{a}{3b}\Big)^2M_{13} + 9\Big(- 1 + \frac{a}{3b}\Big)^2M_{\bar 1 3} + 2\frac{a^2}{b^2}M_{32}\Big]\nonumber\\
&= 10[M_{21} + M_{2\bar 1} + 3M_{23} + 4M_{22}]\nonumber\\
&= 10[M_{31} + M_{3\bar 1} + 2M_{32} + 6M_{33}]\label{eq:gut2}.
\end{align} 

\subsection{Bringing it all together}

We summarize Sections \ref{sec:alg}, \ref{sec:gauge}, \ref{sec:anomalies} and \ref{sec:gut} by saying that for any finite spectral triple whose
\begin{itemize}
	\item[C.1] Lie algebra corresponding to the gauge group is the one of the Standard Model;
	\item[C.2] particle content contains at least one copy of each of the Standard Model particles; 
	\item[C.3] particle content is (chiral gauge) anomaly free;
	\item[C.4] three coupling constants satisfy the GUT relation \eqref{eq:gutrel};
\end{itemize}
the algebra is $\A = M_3(\com) \oplus \Qu \oplus \com$ and the multiplicities of the fermions are constrained by relations \eqref{eq:SManom} and \eqref{eq:gut2} (with $a$ and $b$ appearing in this last relation determined by \eqref{eq:ab}). \\

The readers may ask themselves how and to what extent this approach distinguishes itself from the conventional one, i.e.~the non-NCG approach to (beyond the) SM physics; what more does the former offer compared to the latter? In our opinion there are two main differences. The first is the link between the value for the coupling constants and the Hilbert space ---making the existence of the GUT-relation a consequence of the particle content. Secondly, the demand for the gauge group to be that of the Standard Model made it possible to determine the charges of the featured particles in terms of the powers $a$ and $b$, changing the role of the demand of anomaly cancellations to determining multiplicities. (Both differences are a fruit of the meticulous path from the principles of NCG to the particle content of the Standard Model.)\\

It might be worthwhile to explicate how this article relates to several others having a similar approach, most notably \cite{ISS03,JS05,JSS05,JS08,JS09}. The main differences are that this is the first time that the GUT-relation is explicitly used for constraining the multiplicities, that the demand for the gauge group had not been articulated before in terms of the content of the Hilbert space. Furthermore, this analysis does not regard the question what finite Dirac operators are allowed, so we have no demands related to this. Finally, the articles \cite{ISS03,JS05,JSS05} consider finite spectral triples of KO-dimension $0$, instead of $6$.

\section{Solutions of the constraints}

In the following sections we investigate what the above constraints can tell us about some extensions of the Standard Model. But first we will employ these constraints to recover the latter.

\subsection{SM and extensions thereof}\label{sec:sm}

Let $\A = \mathbb{C} \oplus \mathbb{H} \oplus M_3(\mathbb{C})$ and let $\H$ be such that it only contains representations that are present in the Standard Model: $\mathbf{2}_L \otimes \mathbf{1}^o$, $\overline{\mathbf{1}}_R\otimes \mathbf{1}^o$, $\mathbf{2}_L \otimes \mathbf{3}^o$, $\mathbf{1}_R \otimes \mathbf{3}^o$, and $\overline{\mathbf{1}}_R\otimes \mathbf{3}^o$ (representing $(e_L, \nu_L)$, $e_R$, $(u_L, d_L)$, $u_R$ and $d_R$ respectively) and their conjugates. We leave the possibility open for a right-handed neutrino $\mathbf{1}_R \otimes \mathbf{1}^o$. This implies that $\H$ is characterised by a $6$-tuple 
\begin{align*}
 	(M_{11}, M_{13}, M_{\bar 11}, M_{\bar 13}, M_{21}, M_{23})\in \mathbb{N}^6.
\end{align*}
The subscript $L$ ($R$) refers to the value $1$ ($-1$) of $\gamma_F$ on the particular representation.

\begin{lem}\label{lem:solSM}
	Upon demanding C.1 -- C.4 for this spectral triple, the only solution for the multiplicities is:	
	\begin{align*}
		(M_{11}, M_{13}, M_{\bar 11}, M_{\bar 13}, M_{21}, M_{23}) = (M, M, M, M, M, M) 
	\end{align*}
	with $M \in \mathbb{N}$, the number of \emph{generations} or \emph{families}.
\end{lem}
\begin{proof}
From combining \eqref{eq:SManoma} to \eqref{eq:SManomd} we already infer that
\begin{align}
	M_{13} &= M_{\bar13} = M_{23},& M_{\bar 1 1} &= M_{21} = \frac{a}{b}M_{13} \label{eq:smanom2}.
\end{align}
Inserting these equations into the demand for the GUT-relation \eqref{eq:gut2} (but discarding all non-SM representations), gives:
\begin{align*}
2\bigg(9 + 9 \frac{a}{b} + 2\frac{a^2}{b^2}\bigg)M_{13} &= 40M_{13} = \bigg(30 + 10\frac{a}{b}\bigg)M_{13}
\end{align*}
yielding $a = b$, i.e.~the hypercharges are those of the SM and all multiplicities equal each other, except for $M_{11}$ (the right-handed neutrino), which is still unrestricted. Now from the expressions for $a$ and $b$ we find:
\begin{align*}
	2M_{11} + 2M_{12} = 4M_{21},
\end{align*}
establishing the result. 
\end{proof}
Note that according to the previous Lemma, a right-handed neutrino in all generations is necessary from our point of view. That the Standard Model with three generations and an equal number of right-handed neutrinos is at odds with the orientability axiom, was already noted by C.~Stephan \cite{ST06}. For future convenience we succinctly write the solution to the previous Lemma as:
\begin{align*}
	\H_{\mathrm{SM}}' &:= \big(\mathcal{E} \oplus \mathcal{E}^o\big)^{\oplus M},\qquad \mathcal{E} = (\mathbf{2}_L \oplus \mathbf{1}_R \oplus \overline{\mathbf{1}}_R) \otimes (\mathbf{1} \oplus \mathbf{3})^o,
\end{align*}
where we added the accent because the bare SM in fact has $M = 3$ and does not contain any right-handed neutrinos.\\

We can try to get the most out of the constraints we have derived. We will first focus our attention on the only non-adjoint representation that is absent in the SM --- given that we only use $\mathbf{1}$, $\overline{\mathbf{1}}$, $\mathbf{2}$ and $\mathbf{3}$:
\begin{lem}
	An extension of the SM with only a certain (non-zero) number of copies of $\mathbf{2} \otimes \overline{\mathbf{1}}^o$ (and its conjugate) is possible, provided it is of negative chirality. In that case we have for the multiplicities:
	\begin{align*}
		M_{21} + M_{2\bar1} &= M_{23} = M_{1\bar1} = M_{13} = M_{\bar 13}, & 
		M_{11} &= 3M_{2\bar 1} + M_{21}.
	\end{align*}
\end{lem}
\begin{proof}
	Since the representation under consideration is hypothetical, we do not know whether it is right- or left-handed, but we can try either possibility. We get a $\pm 2M_{2\bar1}$ extra in \eqref{eq:SManoma}, \eqref{eq:SManomc} and \eqref{eq:SManom2}, where the sign depends on the chirality. \\

Taking the latter to be positive (i.e.~it has the same chirality as $M_{21}$), we find from anomaly cancellation that
$
	M_{23} = M_{31} = M_{\bar11}, \frac{a}{b}M_{13} = M_{\bar 11}.
$
 Together with the demand $G_2 = G_3$ this gives 
\begin{align*}\bigg(1 - \frac{a}{b}\bigg)M_{21} = \bigg(-1 + \frac{a}{b}\bigg)M_{2\bar1}.\end{align*}The other GUT-demand $3/5G_1 = G_2$ solves $a/b = -4 \lor 1$. Using all relations between multiplicities, the first solution demands all multiplicities to vanish, the second solution only sets $M_{2\bar1}$ to be zero and we are back at the SM with right-handed neutrinos. 

Using the other value for the chirality, $M_{21}$ and $M_{2\bar1}$ enter all relations in the same way, except in \eqref{eq:ab} ---whose previous use was to determine $M_{11}$ since that is the only constraint it appears in. This means that we cannot exactly solve all multiplicities. Instead we get the results as stated in the Lemma above.
\end{proof}
Looking at what we observe in particle experiments, the above Lemma suggests that either $\mathbf{2} \otimes \overline{\mathbf{1}}^o$ is absent after all, or that there is a non-zero $M_{2\bar1}$ implying that all other particles ---except for $\mathbf{2} \otimes \mathbf{1}^o$--- come in at least one more generation than is currently observed. The representation $\mathbf{1} \otimes \mathbf{1}^o$ (the right-handed neutrino) on the other hand, needs to have an even higher multiplicity than the others.

\subsection{Supersymmetric extensions}\label{sec:susy}

In supersymmetry (see \cite{S85}) one fully exploits all space-time symmetries (as captured by the Poincar\'e group) to arrive at theories in which all existing particles have \emph{superpartners}; particles residing in the same representation of the gauge group, but necessarily differing in spin by the amount of $\nicefrac{1}{2}$. In this article we focus on supersymmetric extensions of the SM in which each SM particle has one supersymmetric partner\footnote{The most common of these is the Minimal Supersymmetric Standard Model (MSSM) (cf.~\cite{CEKKLW05}).}. Thus, we want to extend the finite Hilbert space by the \emph{gauginos}, the superpartners of the gauge bosons (the latter corresponding to the components of the algebra), and the higgsinos (the superpartners of the Higgses \cite[\S 2.2]{CEKKLW05}). The former are in the representations $\mathbf{1} \otimes \mathbf{1}^o$, $\mathbf{2} \otimes \mathbf{2}^o$ and $\mathbf{3} \otimes \mathbf{3}^o$ and the latter are in $\mathbf{2} \otimes \mathbf{1}$ and $\mathbf{2} \otimes \overline{\mathbf{1}}^o$. In order for the result to have the right number of degrees of freedom, we need to take two copies of the gaugino representations, both having a different value of the grading. This allows us not only to project onto the physical states of $\H^+$ [c.f.~\eqref{eq:fermionaction}], halving the number of degrees of freedom, but it also allows for the possibility of defining gaugino masses. Having a real structure $J$ makes the higgsinos automatically come with their (charge) conjugates. Since we already have particles in the representation $\mathbf{1} \otimes \mathbf{1}^o$ and $\mathbf{2} \otimes \mathbf{1}^o$ in the SM, we will distinguish between the SM and supersymmetric versions of this representation by putting a tilde above the latter. In the notation introduced above Lemma \ref{lem:solSM} we write
\begin{align}
	(M_{11} + \widetilde M_{11}, M_{13}, M_{\bar 11}, M_{\bar 13}, M_{21} + \widetilde M_{21}, M_{23}, \widetilde M_{2\bar 1}, \widetilde M_{22}, \widetilde M_{33}) \in\mathbb{N}^{9}\label{eq:solmult2}
\end{align}
with $\widetilde M_{21} = \widetilde M_{2\bar 1} = \widetilde M_{11} = \widetilde M_{22} = \widetilde M_{33} = 1$. Written differently:
\begin{align*}
	\H_{\mathrm{MSSM}}' &:=\H_{\mathrm{SM}}' \oplus \H_{\mathrm{gauginos}} \oplus \H_{\mathrm{higgsinos}}
\end{align*}
with
\begin{align*}
		 \H_{\mathrm{gauginos}} &= \big(\mathbf{1} \otimes \mathbf{1}^o \oplus \mathbf{2} \otimes \mathbf{2}^o \oplus \mathbf{3} \otimes \mathbf{3}^o\big)^{\oplus 2} \simeq \big(\com \oplus M_2(\com) \oplus M_3(\com)\big)^{\oplus 2}\\
		 \H_{\mathrm{higgsinos}} &= \mathbf{2} \otimes \overline{\mathbf{1}}^o \oplus \mathbf{2} \otimes \mathbf{1}^o \oplus \mathbf{1} \otimes \mathbf{2}^o \oplus \overline{\mathbf{1}} \otimes \mathbf{2}^o.
	\end{align*}
We then have:
\begin{lem}
	There is no solution \eqref{eq:solmult2} for the finite Hilbert space that satisfies our demands after extending the Standard Model by two copies of the gauginos and a single copy of the higgsinos.
\end{lem}
\begin{proof}
We can proceed in exactly the same way as in Lemma \ref{lem:solSM}, using the demands C.1 -- C.4. Since the Standard Model representations together satisfied the demands, so should separately do the newly added representations. First of all, the gauginos do not cause an anomaly, since each left representation appears both left-handed and right-handed [c.f.~\eqref{eq:anomresult}]. The higgsino in $\mathbf{2} \otimes \mathbf{1}^o$ does cause an anomaly (via the second and fourth relations in \eqref{eq:SManom}), but the other one in $\mathbf{2} \otimes \overline{\mathbf{1}}^o$ ---having the same grading but an opposite hypercharge--- cancels this anomaly again. So the relations \eqref{eq:smanom2} stay intact, reducing the a priori $6$ unknown SM-multiplicities to only one, $M_{13}$. Next, we find for the three GUT-coefficients:
\begin{align*}
 &12 \frac{a}{b}M_{13} + 6\Big(\frac{a}{b}M_{13} + 1\Big) + 6 + 9\Big(1 + \frac{a}{3b}\Big)^2M_{13} + 9\Big(-1 + \frac{a}{3b}\Big)^2M_{13} + 2\frac{a^2}{b^2}M_{13}\\
&\qquad =  \bigg[18 + 18\frac{a}{b} + 4\frac{a^2}{b^2}\bigg]M_{13} + 12, \\
 & 10\Big[ \Big(\frac{a}{b}M_{13} + 1\Big) + 1 + 3M_{13} + 4\Big] =  10\bigg[\bigg(\frac{a}{b} + 3\bigg)M_{13} + 6 \bigg]\\
&\quad\text{and}\\
 & 10[ M_{13} + M_{13} + 2M_{13} + 6] = 10[4M_{13} + 6],
\end{align*}
respectively. From equating the second and third coefficients one can infer that $a = b$. Inserting this solution into the first coefficient, one gets
\begin{align*}
	12M_{13} + 6M_{13} + 6 + 9\bigg(\frac{4}{3}\bigg)^2 + 9\bigg(-\frac{2}{3}\bigg)^2 + 2M_{13} = 40M_{13} + 12
\end{align*} 
i.e.~the GUT-relation cannot be satisfied. Moreover, inserting the extra multiplicities in \eqref{eq:ab2} shows that $a \ne b$, in contrast to what the GUT-relation told us.
%
\end{proof}

Now the previous lemma suggests that the MSSM and NCG are at odds, but there might be models which are not that different from the MSSM that do satisfy all our constraints. We could in principle restore all constraints by adding extra representations compared to the MSSM. In the light of supersymmetry these should all be a superpartner of a scalar particle that enters through a finite Dirac operator. To show that such models exist, we have

\begin{theorem}
The smallest possible extension (in the sense of lowest number of extra representations) of the MSSM that satisfies all four constraints, has six additional representations in $\H_F$. Namely, it is one of the following two possibilities for the Hilbert space $\H_{\mathrm{MSSM}}' \oplus \mathcal{F} \oplus \mathcal{F}^o$:
\begin{align}
\mathcal{F} &= (\mathbf{1} \otimes \overline{\mathbf{1}}^o)^{\oplus 4}\, \oplus\, (\mathbf{1} \otimes \mathbf{1}^o)^{\oplus 2}\label{eq:NCMSSM},
\intertext{where two of the copies of $\mathbf{1} \otimes \mathbf{1}^o$ should have a grading opposite to the other two copies; or}
\mathcal{F} &= (\mathbf{1} \otimes \overline{\mathbf{1}}^o)^{\oplus 2}\, \oplus\, (\mathbf{1} \otimes \mathbf{3}^o)^{\oplus 2}\, \oplus\, ({\bf 1} \otimes {\bf 2})\, \oplus\, ({\bf{\bar 1}} \otimes {\bf 2})
\label{eq:NCMSSM-2}.
\end{align}
\end{theorem}
\begin{proof}
This can be done with a routine computer check on the equations \eqref{eq:SManom} and \eqref{eq:gut2} while letting the multiplicities $M_{11}, M_{1 \bar 1}, M_{13}, M_{\bar 1 3}, M_{12}, M_{\bar 1 2}, M_{23}, M_{22}$ and $M_{33}$ increase.
\end{proof}

\end{document}